\documentclass{xarticle}

\usepackage{graphicx}
\usepackage{overpic}
\usepackage{stmaryrd}
\usepackage{trimclip}
\usepackage{mathtools}
\usepackage{xspace}

\graphicspath{{graphics_included/}}

\usepackage[numbers,sort]{natbib}

\newcommand{\eemph}[1]{\textbf{\textit{#1}}}

\newcommand{\R}{\mathbb{R}}

\newcommand{\bmat}[1]{\begin{bmatrix}#1\end{bmatrix}}
\def\one{\mathbf{1}}
\def\tp{\mathsf{T}}
\let\mathopfont=\mathrm

\newcommand{\range}{\mathop{\mathopfont{range}}}
\let\bl\bigl
\let\br\bigr
\newbox\vcbox
\def\vcent#1{\setbox\vcbox\hbox{#1}\raise -0.5\ht\vcbox\hbox{#1}}

\def\ugn{\lambda}

\def\itoj{{i \shortto j}}
\def\jtoi{{j \shortto i}}

\def\wu{\omega^\text{u}}

\DeclarePairedDelimiter\abs{\lvert}{\rvert}

\def\wss{\omega^\text{ss}}

\def\bwss{{\bar\omega}^\text{ss}}

\newcounter{edremcounter}
\makeatletter
\DeclareRobustCommand{\shortto}{\mathrel{\mathpalette\short@to\relax}}
\newcommand{\short@to}[2]{\mkern2mu
  \clipbox{{.5\width} 0 0 0}{$\m@th#1\vphantom{+}{\shortrightarrow}$}}
\makeatother

\usepackage{etoolbox}
\apptocmd{\sloppy}{\hbadness 10000\relax}{}{}

\newcommand{\bittide}{bittide\xspace}
\def\ss{{\text{ss}}}
\def\off{{\text{off}}}
\def\tbeta{\tilde{\beta}}
\def\ttheta{\tilde{\theta}}

\def\eat#1{}

\DeclareMathOperator{\sign}{sign}
\DeclareMathOperator{\src}{src}
\DeclareMathOperator{\dst}{dst}

\begin{document}

\title{Buffer Centering for \bittide Synchronization via Frame Rotation}

\author{%
  Sanjay Lall\footnotesymbol{1}
  \and  Tammo Spalink\footnotesymbol{2}}

\note{Preprint}

\maketitle

\makefootnote{1}{S. Lall is with the Department of Electrical
  Engineering at Stanford University, Stanford, CA 94305, USA, and is
  a Visiting Researcher at Google DeepMind.
  \texttt{lall@stanford.edu}\medskip}

\makefootnote{2}{Google DeepMind.}

\begin{abstract}
  Maintaining consistent time in distributed systems is a fundamental
  challenge. The \bittide system addresses this by providing logical
  synchronization through a decentralized control mechanism that
  observes local buffer occupancies and controls the frequency
  of an oscillator at each node. A critical aspect of \bittide's
  stability and performance is ensuring that these elastic buffers
  operate around a desired equilibrium point, preventing data loss due
  to overflow or underflow. This paper introduces a novel method for
  centering buffer occupancies in a \bittide network using a technique
  we term \emph{frame rotation.} 
  We propose a control strategy utilizing a directed
  spanning tree of the network graph. By adjusting the frequencies of nodes
  in a specific order dictated by this tree, and employing a pulsed
  feedback controller that targets the buffer occupancy of edges
  within the spanning tree, we prove that all elastic buffers in the
  network can be driven to their desired equilibrium. This ordered
  adjustment approach ensures that prior centering efforts are not
  disrupted, providing a robust mechanism for managing buffer
  occupancy in \bittide synchronized systems.
\end{abstract}

\section{Introduction}

In distributed computing, maintaining a consistent sense of time
across independent machines presents a fundamental challenge.
Traditional approaches often rely on physical clock distribution or
software protocols to keep local clocks aligned with wall-clock
time. However, these methods can be expensive, introduce asynchrony
with performance consequences, and become impractical at data-center
scales. The \bittide system addresses some of these 
limitations~\cite{spalink_2006} by obviating the need for physical clock
distribution or strict adherence to wall-clock time.

The core innovation of \bittide lies in providing applications with a
notion of time which is \emph{logically synchronized} between nodes~\cite{ls,kenwright2024,prasad}.
This is achieved through a decentralized control mechanism where each
node adjusts its frequency based on observed communication with its
neighbors, which in turn allows construction of a synchronous logical clock
that is unaffected by variations in the underlying physical clock
frequencies.  The \bittide system establishes a shared logical time
across the system, which may be fully disconnected from physical
wall-clock time, allowing logical time-steps to vary in physical
duration both over time and between nodes.  However, from the
perspective of applications running on the system, the behavior is
identical to that of a system with a single shared physical clock.  By
coordinating actions using this logical time, the need to reference
physical time is eliminated.

The decentralized nature of \bittide's synchronization mechanism
enables the construction of large-scale systems which behave as if
they are perfectly synchronized, typically very difficult or
prohibitively expensive to achieve using other methods.

Unlike overlaying synchronization information onto asynchronous
communication layers, which can lead to high communication overhead
and limited accuracy, \bittide leverages the low-level data flows
inherent in serial data links for synchronization. Notably, the
synchronization mechanism requires no additional communication
overhead, as the continuous data exchange at the physical layer
provides a direct feedback signal to the control system.  This allows
for accurate logical synchronization even with an underlying substrate
that is only approximately synchronized.
The \bittide mechanism operates at Layer 1 (physical level) of the OSI
network model, and synchronization occurs
with each node observing only local buffer occupancy levels
associated with links connecting to network neighbors.

Data is transmitted in fixed-size frames and incoming
frames at each node are placed in per-link elastic buffers. A crucial aspect of the
\bittide mechanism is that whenever a frame is removed from the head of
an elastic buffer, a new frame is sent on each outgoing link. In
systems with multiple neighbors, frames are sent simultaneously on all
outgoing links in discrete lockstep.  The oscillator at each node
drives both the processor clock and the network, ensuring that the
lockstep behavior of the network induces a similar behavior in the
processors.

The number of frames in each elastic buffer is measured locally at
each node, and the oscillator frequencies at the nodes are adjusted.
This decentralized control scheme is responsible for keeping frequencies
aligned and ensuring the buffers neither overflow nor underflow. 

In this paper we describe a new method for centering the buffer
occupancies in a \bittide network.  We discuss \emph{frame rotation}, a
method by which a controller may adjust the buffer occupancies in a
\bittide system.  The name alludes to the balance of frames in the 
system being 'rotated' between nodes, taking advantage of a 
\bittide property that total frame counts for every cyclic path in 
the network remain constant.

\section{Notation and preliminaries}

The network model for \bittide used here is a directed graph $\mathcal
G$ with $n$ nodes and $m$ edges, with vertex set $\mathcal{V} =\{1,\dots,n\}$
and edge set $\mathcal{E} =\{1,\dots,m\}$. We will refer to edges interchangeably either
by a source destination pair  $i \to j$ or by edge number $k\in\mathcal{E}$.
The graph has no self loops. The
incidence matrix is $B = S- D$ where $S \in\R^{n \times m}$ is the
source incidence matrix, given by
\[
S_{ie} = \begin{cases}
  1 & \text{if  node $i$ is the source of edge $e$} \\
  0 &\text{otherwise}
\end{cases}
\]
and $D \in\R^{n \times m}$ is the destination incidence matrix 
\[
D_{ie} = \begin{cases}
  1 & \text{if node $i$ is the destination of edge $e$} \\
  0 &\text{otherwise}
\end{cases}
\]
For convenience, we use $\one$ to denote the vector of all ones. Let
$E^j\in\R^{m\times m}$ be the matrix whose entries are all zero except
for $E^j_{jj} = 1$.

We assume the graph is \eemph{strongly connected} or
\eemph{irreducible}, meaning that there exists a directed path in both
directions between any pair of vertices.  Given the graph, suppose
$A\in\R^{n\times n}$ is a nonnegative matrix with sparsity pattern
corresponding to the adjacency, so that that $A_{ij} > 0$ if $i \to j$
is an edge, and $A_{ij}=0$ if $i \neq j$ and $i\to j$ is not an
edge. The matrix $A$ is called irreducible if the corresponding graph
is irreducible, irrespective of the entries on the diagonal.

A matrix $A\in\R^{n\times n}$ is called \eemph{Metzler} if
$A_{ij}\geq0$ for all $i \neq j$, and it is called a \eemph{rate
  matrix} if in addition each of its rows sums to zero. If $A$ is both
Metzler and irreducible, then from the standard Perron-Frobenius
theory there is an eigenvalue $\lambda_\text{metzler}$ which is real,
and which has corresponding positive left and right eigenvectors.  All
other eigenvalues $\lambda$ satisfy $\Re(\lambda) <\lambda_\text{metzler}$.
If $A$ is a rate matrix, then $e^A$ is a stochastic matrix.
In particular,~$DB^\tp$ is a rate matrix with the sparsity of $\mathcal G$
and hence it is irreducible iff $\mathcal G$ is.

Let $\mathcal{T} \subset \mathcal{E}$ be an outward directed spanning
tree. Any spanning tree will do. Let the root node
be~$r\in\mathcal{V}$. There is a natural partial ordering on edges in
$\mathcal{T}$ induced by the tree, where two edges are defined to
satisfy $f \prec g$ if there is a directed walk of non-zero length
from $\dst(f)$ to $\dst(g)$ in the tree.  The actions of our centering
algorithm will follow this ordering.

\section{Model}

A model for the \bittide system on an undirected graph, called the
\emph{abstract frame model}, is developed in~\cite{bms}. That model is
frame accurate, in that it predicts the precise location of every
network frame in the system. For control, we make use of an
approximate differential equation model based on several simplifying
assumptions, including a fluid approximation, and replacement of the
discrete-time control with continuous-time control; for a discussion
of this approximation see~\cite{res}.  Here we consider the directed
graph case, which is a minor change, and build on that model.
Following~\cite{reset} we will simplify the dynamic model by assuming
that the latencies $l_\jtoi$ are zero. Comparison with both hardware and more
detailed simulation of the abstract frame model
which includes latency and individual frames 
have been performed in~\cite{reset,bms,qbay}, and so we do not address that here.
The model is as follows.
\begin{equation}
  \label{eqn:model}
  \begin{aligned}
    \dot\theta_i(t) &= \omega_i(t) \\
    \omega_i(t) &= \wu_i + c_i(t)\\
    \beta_\jtoi(t) &= \theta_j(t) - \theta_i(t) + \ugn_\jtoi\\
    y_i &=  \sum_{j \mid j \to i} (\beta_\jtoi - \beta^\text{off}_\jtoi) 
  \end{aligned}
\end{equation}
Here $i,j \in\mathcal{V} = \{1,\dots,n\}$ refer to graph vertices, and
$\theta_i$ is the clock phase at node $i$, which evolves with
frequency $\omega_i$.
This frequency is the sum of two
terms, the first is $\wu_i$, the \emph{uncontrolled} frequency, which
is frequency of the oscillator without any control, and the
\emph{correction}, $c_i$, a frequency adjustment which is chosen by
the controller. 

At node $i$ there is an elastic buffer associated with each incoming
link $j\to i$, which contains $\beta_\jtoi$ frames. The controller at
node $i$ measures the sum of the occupancies of the buffers, denoted
$y_i$.  The constant $\lambda_\jtoi$ is a property of the link, and
the constant $\beta^\text{off}_\jtoi$ is known and set on
initialization. The base frequencies of the oscillators $\wu_i$ are
constant but unknown.  At each node $i$ the controller measures $y_i$
and chooses the frequency correction $c_i$. In practice this is
sampled, but in this paper we will assume that the sampling is fast
enough that the controller can be treated as continuous-time.

The fundamental dynamics of \bittide are as follows. At each node $i$
there is an oscillator of frequency $\omega_i$, which drives the clock
phase $\theta_i$. We refer to the local clock \emph{ticks} at node $i$
as those times $t$ at which $\theta_i(t)$ is an integer. Nodes are
connected by network links, corresponding to the edges of the
graph. At each node there is one FIFO buffer for each incoming link,
and incoming data frames are stored in the buffer. With each local
clock tick, a frame is removed from all of the buffers, and passed to
the processor at that node. In addition, with each local tick, on each
outgoing link, a new frame is sent by the processor. As a result of
these dynamics, the number of frames in the buffer at node $i$
corresponding to the link from node $j$, denoted $\beta_\jtoi$, is
approximately given by~\eqref{eqn:model}. An explicit derivation
of this is given in~\cite{bms}.

At each node there is a controller, which observes the occupancies of
each of the buffers at that node. It uses this observation in order to
set the oscillator correction~$c_i$. The idea is that, if the
oscillator at node $i$ is slower than that of its neighbors, then it
will send frames less frequently than it receives them, and its
buffers will start to fill up; conversely, too fast and it's buffers
will drain. This motivates a controller of the form
\begin{equation}
  \label{eqn:ctrl}
  c_i(t) = k \sum_{j \mid j \to i} (\beta_\jtoi - \beta^\text{off}_\jtoi) 
\end{equation}
Here $\beta^\text{off}$ is the desired equilibrium point, which
is usually the midpoint of the buffer. By adjusting the correction, the nodes
must not only manage to ensure that all nodes tick at approximately
the same frequency, but they must also ensure that all of the buffer
occupancies remain close to their corresponding $\beta^\text{off}$.
The latter is particularly important, because buffer overflow or
underflow will cause running applications to lose data, and is a fatal
error. The controller~\eqref{eqn:ctrl} is a proportional-plus-offset
controller. We will make use of the constant term $q_i$ in the
controller to adjust the buffer equilibrium points, as discussed
below.

Using the incidence matrices
$B\in\R^{n\times m}$ and $D\in\R^{n\times m}$, we can write
\begin{equation}
  \label{eqn:model2}
  \begin{aligned}
    \dot\theta(t) &= \omega(t) \\
    \omega(t) &= \wu + c(t)\\
    \beta(t) &= B^\tp \theta(t) + \ugn\\
    y &=  D (\beta - \beta^\text{off}) 
  \end{aligned}
\end{equation}
Following~\cite{reset}, 
we make the following assumption about the system boot, referred to as
\emph{feasibility}.
\begin{asm}
  \label{asm:feas}
  There exists some time $t^0$ at which the buffer occupancy $\beta(t^0) = \beta^\text{off}$.
\end{asm}
We now define for convenience the following choice of normalized coordinates.
\begin{equation}
  \label{eqn:normalization}
  \ttheta(t) = \theta(t)  - \theta(t^0) \qquad
  \tbeta(t) = \beta(t) - \beta^\text{off}
\end{equation}
We will define for convenience the directed Laplacian matrix
\[
Q = DB^\tp
\]
This gives the following.
\begin{prop}
  \label{prop:model}
  Under Assumption~\ref{asm:feas}, the dynamics~\eqref{eqn:model2} are equivalent
  to the following
  \begin{align*}
    \dot{\ttheta}(t) &= \wu + c(t) \\
    \tbeta(t) &= B^\tp \ttheta(t)\\
    y &= D \tbeta(t)
  \end{align*}
  together with the boundary condition $\ttheta(t^0)=0$.
\end{prop}

\paragraph{Properties of the Laplacian.}
If the graph is irreducible then the matrix $Q$ is an
irreducible rate matrix. Let $z>0$ be it's Metzler eigenvector,
normalized so that $\one^\tp z = 1$. Let the eigendecomposition of $Q$ be
$QT = TD$, then we have
\[
D = \bmat{0 & 0 \\ 0 & \Lambda}
\quad
T = \bmat{\one & T_2}
\quad
T^{-1} = \bmat{z^\tp \\ V_2^\tp}
\]
Define the matrix
\[
Q^\ddag = T \bmat{0 & 0 \\ 0 & \Lambda^{-1}}T^{-1}
\]
which satisfies $Q Q^\ddag Q = Q$ and so is a generalized inverse of $Q$,
and the associated projector $W = I - Q^\ddag Q$. Explicitly, we have $W =
\one z^\tp$. It is immediate that $WQ = QW = WQ^\ddag = Q^\ddag W =0$. Further,
since $We^{Qt} = W$ we have
\[
e^{Qt} = W + Q^\ddag Q e^{Qt}
\]
and so
\[
e^{Qt} = \frac{d}{dt}\bbl( Wt + Q^\ddag e^{Qt} \bbr)
\]

\paragraph{Behavior with proportional control.}

Previous work has considered use of proportional and proportional-integral control~\cite{bms,reset}.
Our approach in this paper builds on this, applying two stages of control. In the first stage,
we will use proportional control, and subsequently the control will switch to a sequence
of controllers which use a pulsed input. The closed-loop dynamics for $\ttheta$ become
\[
\dot{\ttheta}(t) = kQ \ttheta(t) + \wu
\]
The matrix $Q$ is not Hurwitz, so this dynamics is
not strictly stable and $\theta(t) \to \infty$ as $t \to \infty$. This
is expected, since $\theta$ is the clock phase and must grow without bound.
Despite this, the buffer occupancy converges. We will need the solution
of this system, which we state here for convenience.
\begin{prop}
  Suppose $\ttheta$ and $\tbeta$ satisfy the dynamics of Proposition~\ref{prop:model},
  and the controller is given by $c(t) = k y(t)$ where $k>0$.  Then
  \[
  \ttheta(t) = \bbl(
  Wt + k^{-1} Q^\ddag (e^{kQt} - I)
  \bbr) \wu
  + e^{kQt} \ttheta(0)
  \]
  and hence as $t\to\infty$ we have $\theta(t) \to \theta^\ss$ and $\beta(t) \to \beta^\ss$ where
  \[
  \omega^\ss = W \wu \qquad \beta^\text{ss} = -k^{-1} B^\tp Q^\ddag \wu
  \]
\end{prop}
\begin{proof}
  The proof follows from the standard variation-of-constants formula for linear dynamical systems.
  Note that we make use of the property that $Q$ is Metzler to conclude
  \[
  \lim_{t\to\infty} e^{Qt}  = W
  \]
  along with $Q^\ddag W =0$ and $B^\tp W = 0$.
\end{proof}
One important observation is that, a proportional controller will
drive the system to an equilibrium point for which all nodes have the
same frequency. This follows from the above expression for $\wss$,
because $W=\one z^\tp$. Define $\bwss = z^\tp\wu$ to be this frequency.

\section{Control of Buffer Occupancy}

The objective of the control system is to ensure that the buffer
occupancies $\beta_\itoj$ are kept to prescribed levels.  To do this,
it is essential that all nodes maintain, on average, approximately the
same frequency, since the buffer occupancy for edge $i\to j$ increases
at a rate proportional to the difference between the frequency of node
$i$ and that of node $j$. If on average node $i$ has a higher
frequency than node $j$, then the buffer occupancy will increase
without bound.

There are several approaches for achieving this. A
proportional-integral controller is used in~\cite{bms}, and a reset
controller is used in~\cite{reframing}. Both of these methods are
effective at ensuring that buffer occupancies are kept close to
$\beta^\text{off}$. One of the difficulties with controlling \bittide
is that any control scheme must be flexible enough to allow nodes to
be be added and removed while the system is operational. Here we
present a method to allow control of buffer occupancies directly. We allow
individual nodes to apply feedback control in such a way as to control
the local buffer occupancies, and show that this achieves a
desirable outcome for all buffer occupancies on the network.

When all nodes are at the equilibrium frequency, the critical
observation is that a single node $i$ can adjust the buffer occupancy
at the elastic buffers for its incoming links. To do this, it
temporarily changes its frequency $\omega_i$, while the other nodes
keep their frequencies constant. This will cause the elastic buffers
at node $i$ to drain. It can therefore set one of the elastic buffer
occupancies to the midpoint, simply by increasing or decreasing its
frequency for a short amount of time.

\subsection{Example: triangular network}
\label{s:example}
Consider the system with three nodes illustrated in
Figure~\ref{fig:three}. In this system, the nodes start at frequencies
$\wu$ and rapidly converge to a common frequency. At approximately
$t\approx 400e6$, node 2 reduces its frequency, as shown in
Figure~\ref{fig:threefreq}. With this reduced frequency, the elastic
buffers at node 2 start to fill, as can be seen in
Figure~\ref{fig:threemocc}. Node 2 observes the occupancy of the
elastic buffer $\beta_{1\shortto 2}$, and when it reaches the
midpoint, it resets its frequency to the equilibrium value.

\begin{figure}[ht!]
  \centerline{\begin{overpic}[width=0.28\linewidth]{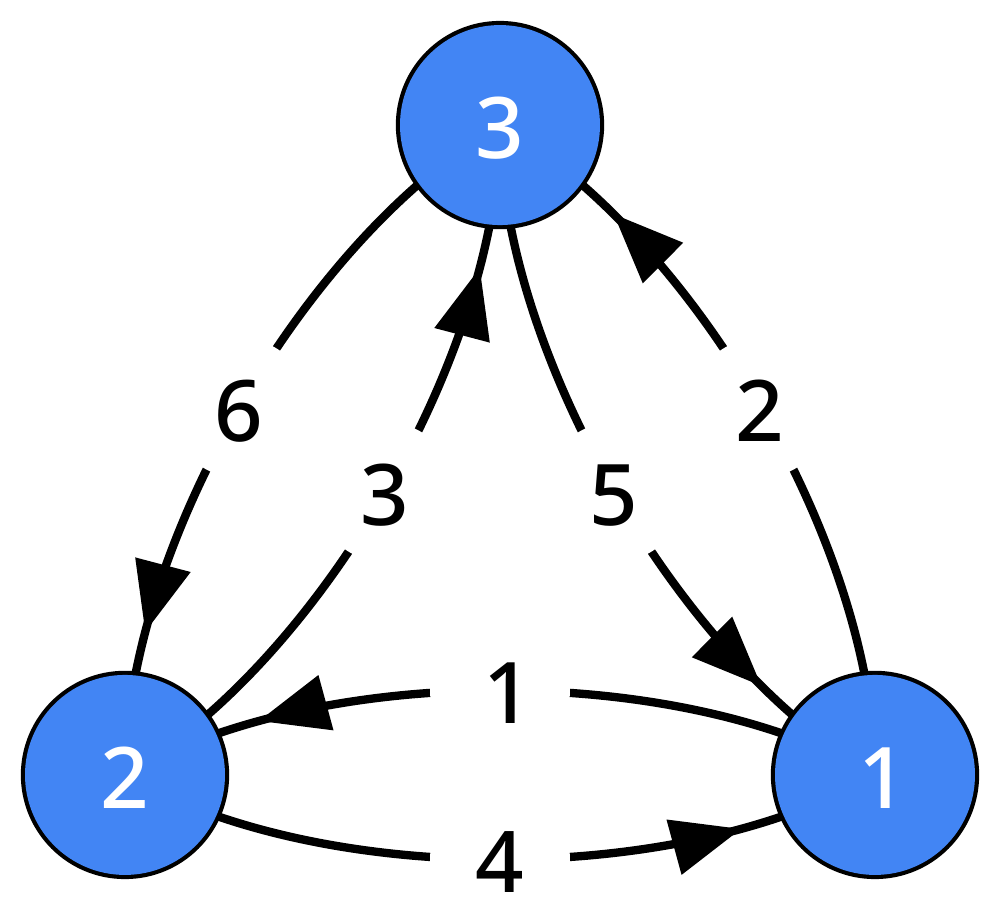}\end{overpic}}
  \captionsetup{margin=30pt}
  \caption{Graph used for example in \S\ref{s:example}}
  \label{fig:three}
\end{figure}

This strategy achieves the immediate goal of centering
$\beta_{1\shortto 2}$.  By chance, it also centers the buffer
occupancy of $\beta_{2\shortto 1}$.  In addition, the buffer
occupancies of the elastic buffers associated with the other edges
either incoming to node 2 or outgoing from node 2 are also affected.
This suggests that the strategy of each node simply successively
centering its elastic buffers will not work, since each nodes actions
will potentially de-center the changes which happened before.  This
can certainly happen; for example, if node 1 now increases its
frequency for a short period in order to center the buffer occupancy
of edge 5, this will then cause the elastic buffers of edges 1 and 4 to become
de-centered again. This is shown in Figure~\ref{fig:threepost},
where edge 5 is controlled starting at time $t\approx 700e6$.

\begin{figure}[ht!]
  \centerline{\begin{overpic}[width=\linewidth]{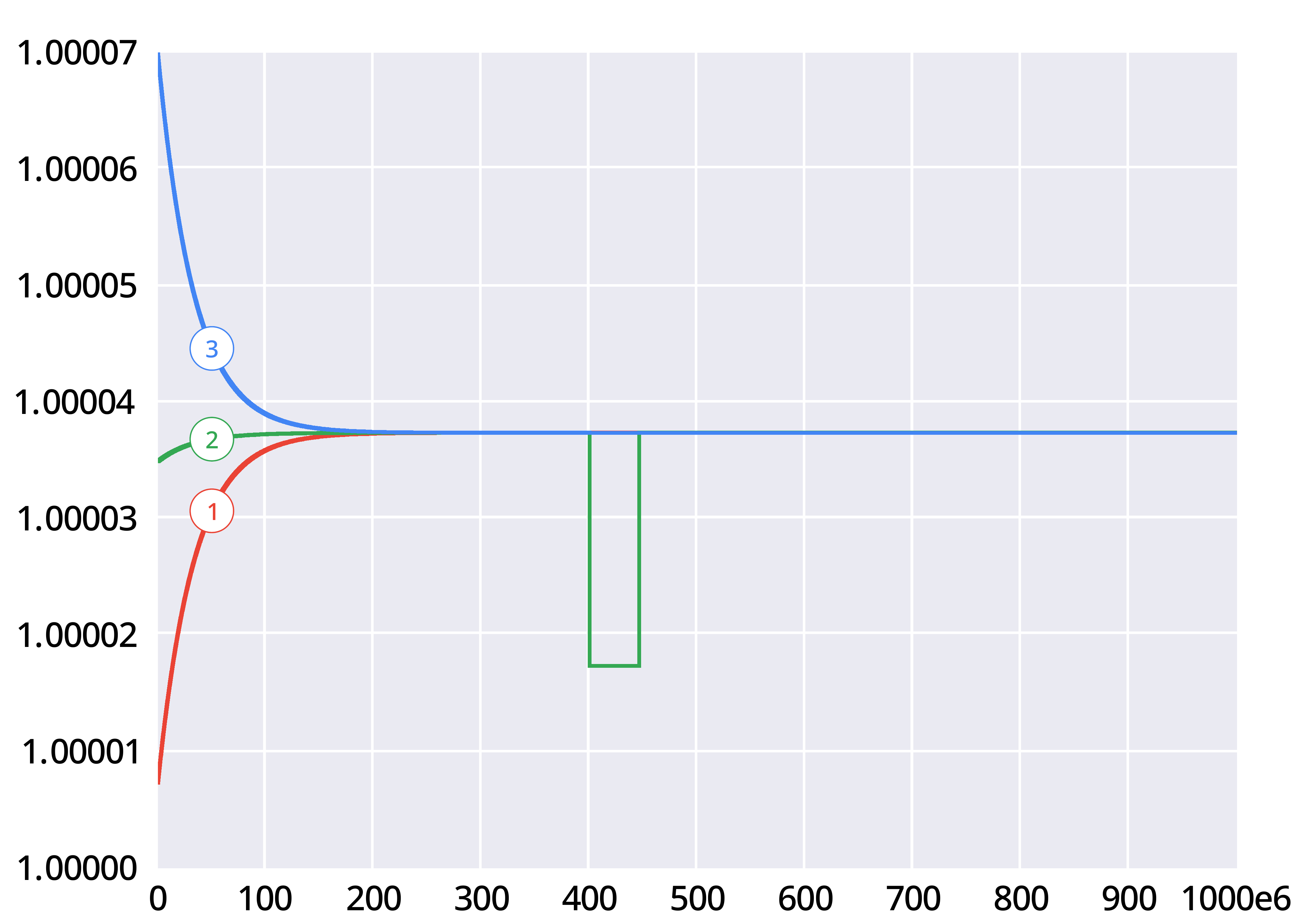}
      \put(13,62){\small$\omega$}
      \put(91,5){\small$t$}
    \end{overpic}}
  \caption{Frequency behavior as a function of time for the system in Figure~\ref{fig:three}.}
  \label{fig:threefreq}
\end{figure}

\begin{figure}[ht!]
  \centerline{\begin{overpic}[width=\linewidth]{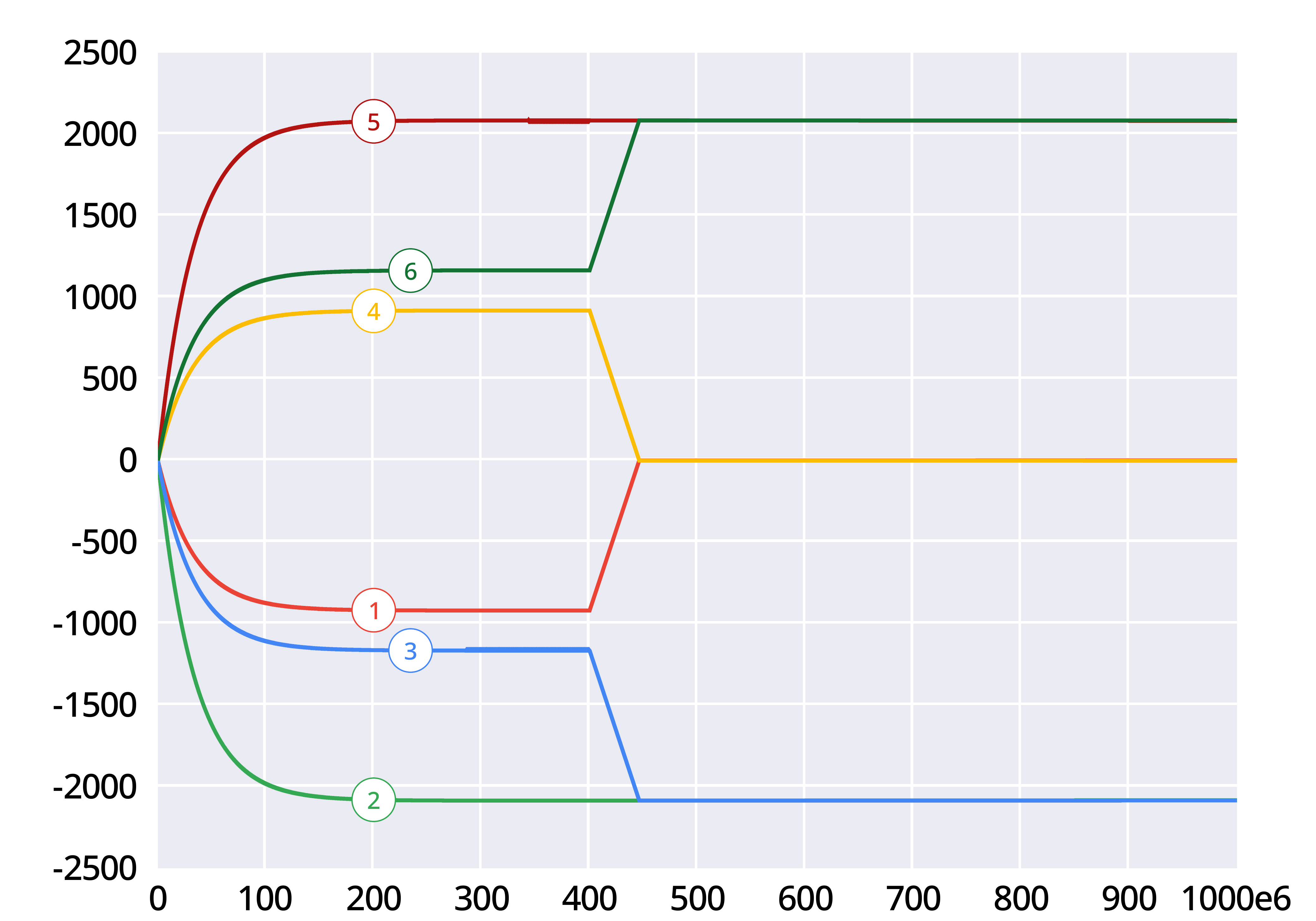}
      \put(13,61){\small$\tbeta$}
      \put(91,5){\small$t$}
    \end{overpic}}
  \caption{Relative buffer occupancies for the system in Figure~\ref{fig:three}.}
  \label{fig:threemocc}
\end{figure}

\begin{figure}[ht!]
  \centerline{\begin{overpic}[width=\linewidth]{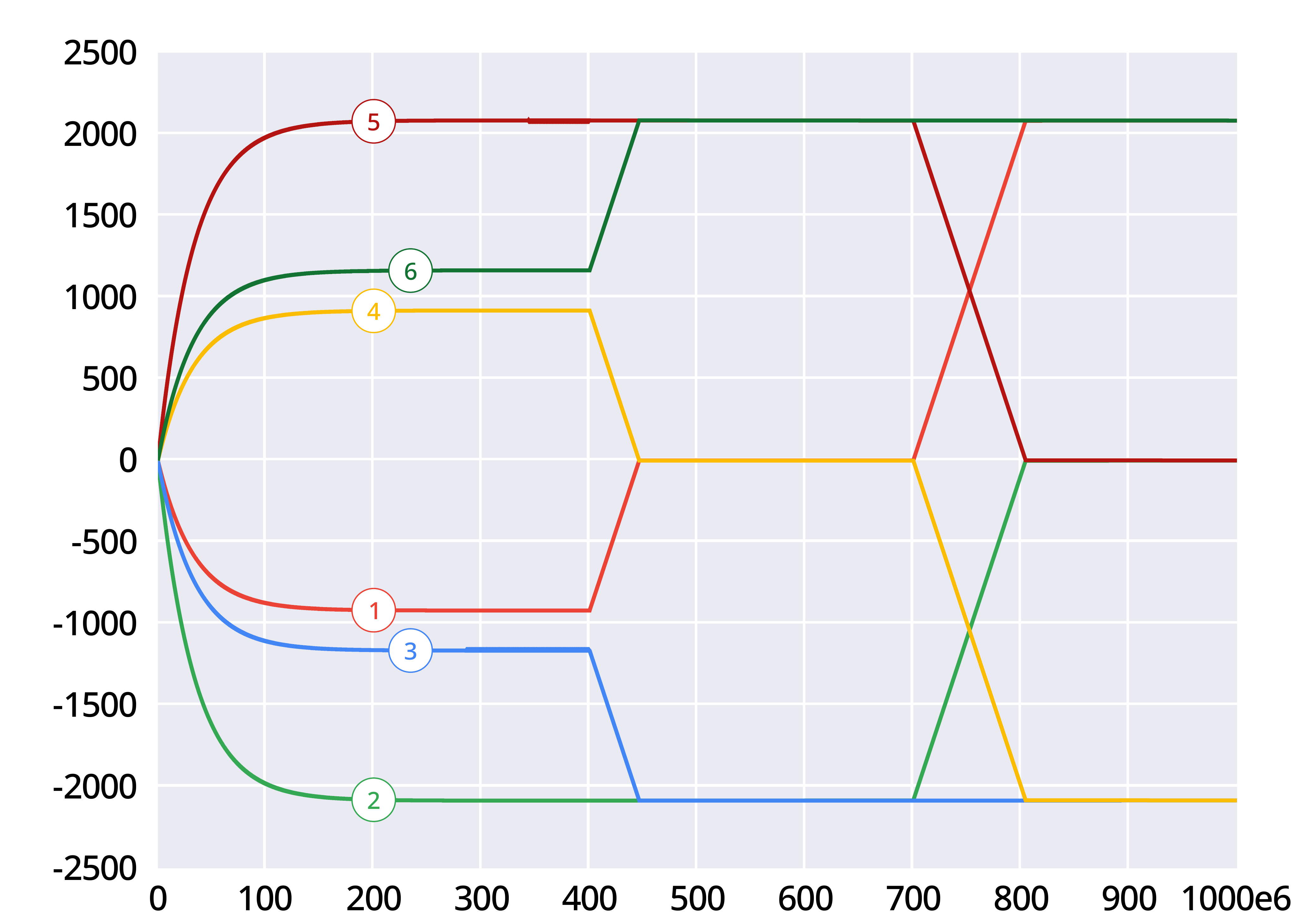}
      \put(13,61){\small$\tbeta$}
      \put(91,5){\small$t$}
    \end{overpic}}
  \caption{Relative buffer occupancies for the system in Figure~\ref{fig:three} after
  centering by both nodes 2 and node 1.}
  \label{fig:threepost}
\end{figure}

The simulations in this paper were performed with
Callisto~\cite{callisto}, which is a full simulation of the individual
frames in a bittide network, including latency. The buffer occupancies
observed in simulation match those predicted in the theory of this paper,
despite the different levels of modeling fidelity.

\subsection{Example: mesh}
In this paper, we present an approach for solving this problem. Specifically,
we show that there is an ordering in which nodes can apply corrections to
the elastic buffer occupancies, and that after applying this sequence
of corrections, \emph{all} of the buffer occupancies on the graph are centered.

Our approach is as follows. First, we construct a directed spanning
tree on the graph. An example is shown in Figure~\ref{fig:mesh}. All
paths on the spanning tree lead away from the (arbitrary) root node; in this
example, the root node is 2. Each edge has a corresponding elastic
buffer at its destination node.  The proposed control policy must
satisfy the following requirement.  All nodes on the network are
adjusted apart from the root node. If node $i$ is an ancestor of node
$j$ in the tree, and $i$ is not the root, then we must adjust the
frequency of node $i$ before we adjust that of node $j$. This ordering
is the partial ordering corresponding to the spanning tree.

\begin{figure}[ht!]
  \centerline{\begin{overpic}[width=0.65\linewidth]{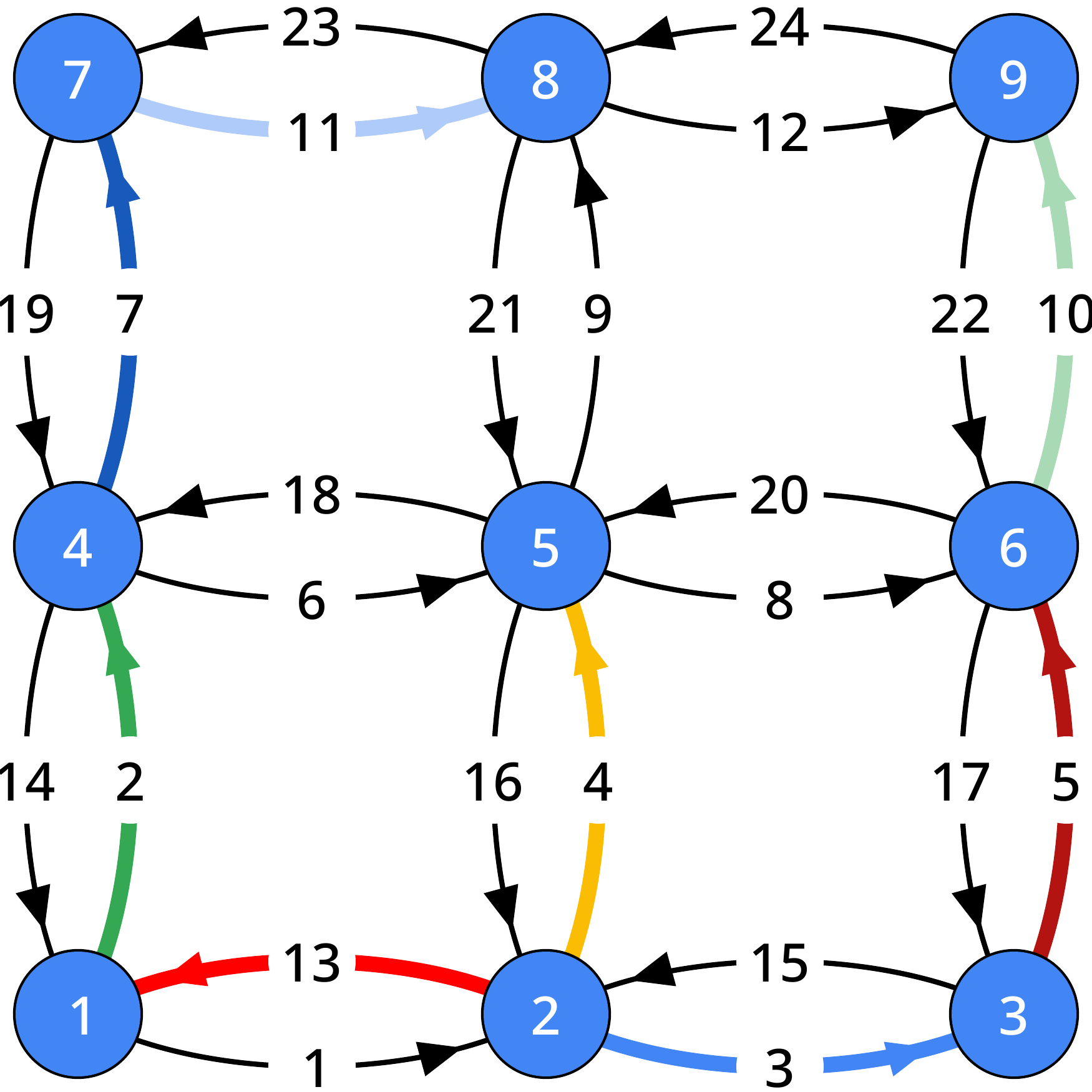}\end{overpic}}
  \captionsetup{margin=30pt}
  \caption{Graph with directed spanning tree}
  \label{fig:mesh}
\end{figure}

For the example in Figure~\ref{fig:mesh}, one possible ordering of
edge adjustments is $(3,5,10,4,13,2,7,11)$. The resulting behavior of
the elastic buffers is shown in Figure~\ref{fig:meshocc} where the
edges are adjusted at times $500, 750, \dots, 2250$.

\begin{figure}[ht!]
  \centerline{\begin{overpic}[width=\linewidth]{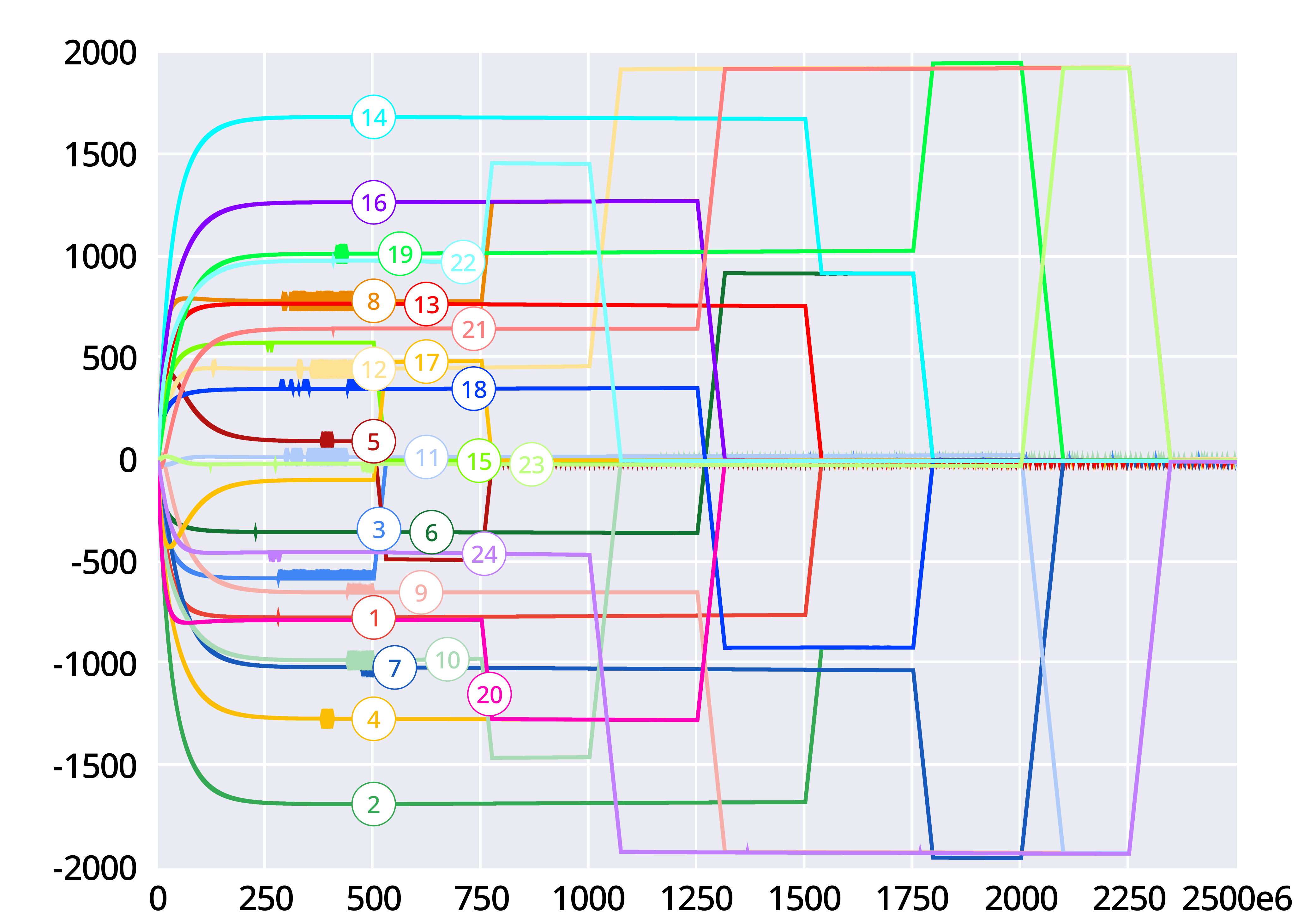}
      \put(13,61){\small$\tbeta$}
      \put(91,5){\small$t$}
    \end{overpic}}
  \caption{Relative buffer occupancies for the system in Figure~\ref{fig:mesh} after
  centering by all nodes in sequence.}
  \label{fig:meshocc}
\end{figure}

\section{Formulation of the controller}

In this section we state precisely the controller that we will use, and
show that it has the desired outcome of reducing all buffer
occupancies to the midpoint; that is $\beta = \beta^\off$, or
equivalently $\tbeta = 0$.

\begin{defn}
  \label{defn:ctrl}
  Suppose $\mathcal{T}$ is an outward directed spanning tree with root
  $r$. Let $g_1,\dots,g_{n-1}$ be an ordering of the edges in
  $\mathcal{T}$ which is consistent with the natural ordering; that
  is, if $g_i \prec g_j$ then $i < j$.  Let $0 < t_1 < t_2 < \dots <
  t_n$ be given, with $t_{i+1} - t_i$ sufficiently large. Let $k$ and
  $k_2$ be positive.
  
  For each node $i\in\mathcal{V}$ we define the controller as follows.
  Given $i$, there is exactly one edge in $\mathcal{T}$  which has
  destination $i$; define $j$ so that $g_j$ is that edge. Then let the controller
  be
  \[
    c_i(t) = \begin{cases}
      k y_i(t) & \text{if } t < t_1 \\[2mm]
      k y_i(t_1) + k_2 \sign(\tbeta_{g_j}(t))  & \text{if }
      \begin{aligned}[t]
        t_j \leq t < t_{j+1} \\
        \text{ and }i\neq r\\[1mm]
      \end{aligned}
      \\
      k y_i(t_1) & \text{otherwise}
    \end{cases}
    \]
\end{defn}

With the controller defined as above, the system is using a proportional controller
$c = ky$ for the time $t<t_1$. We will make the following assumption.

\begin{asm}
  \label{asm:conv}
  We assume that at time $t=t_1$ that the system has converged; that is
  both $\beta(t_1) = \beta^\ss$ and $\omega(t_1) = \omega^\ss$.
\end{asm}

We can now look at the system behavior as a result of using this controller.
First we look at the effect of a single interval.

\begin{lem}
  \label{lem:oneedge}
  Consider the dynamics of Proposition~\ref{prop:model}. Suppose $g\in\mathcal{E}$ is
  an edge with destination vertex $i=\dst(g)$. Suppose
  for $t \in [t_1,t_2]$ the controller is
  \[
  c(t) = q + k_2 \sign(\tbeta_g(t)) e_i
  \]
  where $q = (W-I)\wu$. Let $h=\abs{\tbeta_g(t_1)}/k_2$ and assume
  $t_2 > t_1+h$. Then
  \[
  \tbeta(t_2) = (I + B^\tp D E^g) \tbeta(t_1)
  \]
\end{lem}
\begin{proof}
  From the dynamics, we have
  \begin{align*}
    \dot\tbeta(t) &= BT^\tp \dot\ttheta \\
    &= B^\tp (q + \wu + k_2 \sign \tbeta_g(t) e_i) \\
    &= k_2 \sign \tbeta_g(t) e_i
  \end{align*}
  since $B^\tp W = 0$. Therefore
  \[
  \dot\tbeta_g(t) = -k_2 \sign \tbeta_g(t)
  \]
  which is a scalar on-off feedback system. Hence $\tbeta_g(t_1+h) =
  0$. Let $s=\sign(\tbeta_g(t_1))$ then
  \begin{align*}
    \tbeta(t_2) &= \tbeta(t_1) + \int_{t_1}^{t_2} \dot\tbeta(t)\,dt \\
    &=  \tbeta(t_1) + \int_{t_1}^{t_1+h}  k_2 s B^\tp e_i \, dt \\
    &=  \tbeta(t_1) + k_s s h B^\tp e_i 
  \end{align*}
  Now $hk_2s= e_g^\tp \beta(t_1)$ and so
  \begin{align*}
    \tbeta(t_2) &= (I + B^\tp e_i e_g^\tp) \tbeta(t_1)  \\
    &= (I + B^\tp D E^g) \tbeta(t_1) 
  \end{align*}
  as desired.
\end{proof}

In Lemma~\ref{lem:oneedge}, a critical assumption is that the
controller has access to $q$. This is possible, even though $\wu$ is
not known by the controller, by making use of the equilibrium of the
proportional controller. Specifically, a proportional controller has
an equilibrium such that $k y = (W-I)\wu$. By applying the buffer
adjustments after the system has reached equilibrium, each node $i$
has access to $k y_i$, which is sufficient to apply the result.

\begin{lem}
  \label{lem:step}
  Consider the dynamics of Proposition~\ref{prop:model} and the
  controller of Definition~\ref{defn:ctrl}, and let
  Assumption~\ref{asm:conv} hold.  Then
  \[
  \tbeta(t_{j+1}) = (I + B^\tp D  E^{g_j}) \tbeta(t_j)
  \]
  and in particular
  \[
  \tbeta_{g_j}(t_{j+1}) = 0
  \]
  for all $j=1,\dots,n-1$.
\end{lem}
\begin{proof}
  Using Assumption~\ref{asm:conv}, we have $\dot\ttheta(t_1)=0$ and so
  $ky(t_1) =(W-I)\wu$.  The proof then follows from applying
  Lemma~\ref{lem:oneedge} to the controller in
  Definition~\ref{defn:ctrl}.
\end{proof}

Lemma~\ref{lem:step} shows that in the interval $[t_j,t_{j+1}]$, one
of the elastic buffers is set to zero, and the frequency at the end of
the interval is the same as it was at the start. We now consider the
effect of such a step on the other elastic buffers.
\begin{lem}
  \label{lem:stepbuffers}
  Assume $g\in\mathcal{E}$ is an edge. Let $x\in\R^m$, and
  \[
  z = (I+B^\tp D E^g) x
  \]
  Then $z_{g} = 0$ 
  and for all $l\in\mathcal{E}$ such that $\src(l) \neq \dst(g)$ and $\dst(l) \neq \dst(g)$
  we have $  z_l = x_l$.
\end{lem}
\begin{proof}
  It follows directly from the definitions that
  \[
  (B^\tp DE^g)_{ab}
  = \begin{cases}
    -1 & \text{if $b=g$ and $\dst(a) = \dst(g)$} \\
    1 & \text{if $b=g$ and $\src(a) = \dst(g)$} \\
    0 & \text{otherwise}
  \end{cases}
  \]
  from which the result follows.
\end{proof}

\begin{lem}
  \label{lem:manybuffers}
   Consider the dynamics of Proposition~\ref{prop:model} and the
  controller of Definition~\ref{defn:ctrl}, and let
  Assumption~\ref{asm:conv} hold.  Then for every edge $a\in\mathcal{T}$
  \[
  \tbeta_a(t_n) = 0
  \]
  That is, after the controller is applied for $n-1$ steps, the buffer
  occupancy of all tree edges is zero.
\end{lem}
\begin{proof}
  We have by Lemma~\ref{lem:step} that 
  \[
  \tbeta(t_{j+1}) =  \bl(I + B^\tp D E^{g_j} \br) \tbeta(t_j)
  \]
  By Lemma~\ref{lem:stepbuffers} we have that $\tbeta_{g_j}(t_{j+1})= 0$ for every $j=1,\dots,n-1$.
  We claim that  $\tbeta_a(t_{j+1}) = 0$ for all edges $a \preceq g_j$ for all $j$. The proof
  follows by induction. By  Lemma~\ref{lem:stepbuffers} it holds at $j=1$. Now suppose it holds
  at step $j$. We will show that it holds at step $j+1$.  We have
  \[
  \tbeta(t_{j+2}) =  \bl(I + B^\tp D E^{g_{j+1}} \br) \tbeta(t_{j+1})
  \]
  Now if $a\preceq g_{j+1}$ then $\src(a) \neq j+1$, since the tree
  must be acyclic. Also there is only one edge $a$ such that $\dst(a)
  = \dst(g_{j+1})$, since it is an outward tree, and that edge is $a =
  g_{j+1}$.  So by Lemma~\ref{lem:stepbuffers} we have
  $\tbeta_{g_{j+1}}(t_{j+2}) = 0$. All other edges $a$ have $\tbeta_a(t_{j+2}) = \tbeta_a(t_{j+1})$
  and so using the induction hypothesis this must equal zero.  
\end{proof}

The following result is well-known.

\begin{thm}
  \label{thm:smith}
  Suppose $\mathcal G = (\mathcal V, \mathcal E)$ is a directed graph
  with incidence matrix $B$, and suppose edges $1,\dots,n-1$ form a
  spanning tree.
  Partition $B$ according to
  \[
  B = \bmat{B_{11}  & B_{12} \\ -\one^\tp B_{11} & -\one^\tp B_{12}}
  \]
  then $B_{11}$ is unimodular. Further
  \[
  B =
  \bmat{ B_{11} & 0 \\ -\one^\tp B_{11} & 1}
  \bmat{I & 0 \\ 0 & 0}
  \bmat{I & N \\ 0 & I}
  \]
  where $N = B_{11}^{-1}B_{12}$.
\end{thm}
\begin{proof}
  See for example Theorem~2.10 of~\cite{bapat}.
\end{proof}

\begin{lem}
  \label{lem:range}
  Suppose $\mathcal G = (\mathcal V, \mathcal E)$ is a directed graph
  with incidence matrix $B$, and suppose edges $\mathcal{T} = \{1,\dots,n-1\}$ form a
  spanning tree. Let $B$ be the incidence matrix of the graph and
  suppose $y\in\range(B^\tp)$. If $y_i=0$ for all $i\in\mathcal{T}$ then $u=0$.
\end{lem}

\begin{proof}
  Since $y\in\range(B^\tp)$ we have, by Theorem~\ref{thm:smith}, that there
  exists $x\in\R^n$ such that
  \[
  y = \bmat{I & 0 \\ N^\tp & I}
  \bmat{I & 0 \\ 0 & 0 }
  \bmat{B_{11}^\tp & -B_{11}^\tp \one \\ 0 & 1} x
  \]
  and hence  $y=\bmat{I \\ N^\tp} z$ for some $z\in\R^{n-1}$. Now,
  since $y_i=0$ for all $i\in\mathcal{T}$ we have
  \[
  \bmat{0 \\ \hat{y}} = \bmat{I \\ N^\tp} z
  \]
  where $\hat{y} \in\R^{m-n+1}$. Hence $z=0$ and $\hat{y}=0$. 
\end{proof}

There are several parameters in the controller. The theoretical
requirements are that $k$ and $k_2$ be positive and that the time
intervals $t_{i+1}-t_i$ be sufficiently large. In practice the system
is not sensitive to these choices.  The proportional gain $k$ is
limited in size by the sample rate in a practical implementation.
Another practical consideration is that nodes do not have access to
the exact time $t$. The controller is not dependent on the exact
choice of $t_i$ and so small inaccuracies here do not affect its
behavior.  Other constraints, such as the quantization in the
frequency control mechanism, may also play a role but we do not
analyze that here. The controller also requires determination of a
spanning tree in advance, and the choice of any order consistent with
the tree, which can be determined from the topology via standard
algorithms before the controller is run. An additional feature of the
implementation is that the controller of Definition~\ref{defn:ctrl}
only specifies the control input up to time $t_n$; after this time,
the system is in a relative equilibrium and a simple proportional
controller may be used.

We can now state the main result of this paper. If the controller is
determined using the spanning tree and the elastic buffers of the
edges on this spanning tree are centered by adjusting the frequencies
of the nodes in an ordering consistent with the partial ordering
according to $T$, then all elastic buffers will
be centered. This is stated below.
  
\begin{thm}
  Consider the dynamics of Proposition~\ref{prop:model} and the
  controller of Definition~\ref{defn:ctrl}, and let
  Assumption~\ref{asm:conv} hold. Then
  \[
  \tbeta(t_n) = 0
  \]
\end{thm}
\begin{proof}
  Lemma~\ref{lem:manybuffers} show that $\tbeta_a(t_n)=0$ for
  all $a$ in the spanning tree. Now since $\tbeta = B^\tp \ttheta$
  we know that $\tbeta\in\range(B^\tp)$.   Using Lemma~\ref{lem:range}
  gives the desired result.
\end{proof}

Lemma~\ref{lem:manybuffers} shows that since the orderings proceed
away from the root, no adjustment de-centers any elastic buffer on the
tree preceding it.  Therefore after this process, all of the elastic
buffers on the tree are centered.  Lemma~\ref{lem:range} then provides
the final step, showing that if the elastic buffers on a spanning tree
are all centered, then every elastic buffer is centered.

\section{Conclusions}

This paper presents a novel approach, termed \emph{frame rotation},
for achieving buffer centering in \bittide synchronized networks.

We provide examples and analysis that show how purely local and
uncoordinated buffer adjustments can shift the equilibrium and thereby
de-center other buffers within the system.

To overcome this, we propose a structured control strategy predicated
on the construction of a directed spanning tree of the underlying
graph.  We give a method which carefully orchestrates the frequency
adjustments of individual nodes in an order consistent with this
spanning tree, and employs a pulsed feedback mechanism to adjust the
buffer occupancy of specific tree edges. We prove that this approach
ensures the convergence of all elastic buffers to their desired
equilibrium.

The significance of this work lies in providing a robust and
theoretically grounded method for managing a critical aspect of
\bittide system operation. While our approach necessitates an initial
coordination phase to establish the spanning tree and the processing
order, the subsequent adjustments are performed locally, leveraging
readily available buffer occupancy information.  This methodology
provides a new and reliable technique for controlling the data flows
on \bittide networks. Future work could explore the resilience of this
approach to dynamic network changes, in particular adding and removing
nodes, and investigate methods for a more adaptive or fully
decentralized determination of the spanning tree and processing order.

\bibliographystyle{abbrv}

\end{document}